\documentclass[final]{dmtcs-episciences}

\usepackage[utf8]{inputenc}
\usepackage{amsmath,amssymb}
\usepackage{url}

\usepackage[caption=false]{subfig}
\usepackage{braket}
\usepackage{amsfonts}
\usepackage{amsthm}
\usepackage{fancyhdr}
\usepackage{comment}
\usepackage{float}
\usepackage{mathtools}
\usepackage{xspace}
\usepackage[export]{adjustbox}

\usepackage{diagbox} 

\newtheorem{theorem}{Theorem}
\newtheorem{lemma}[theorem]{Lemma}
\newtheorem{definition}[theorem]{Definition}
\newtheorem{corollary}[theorem]{Corollary}
\newtheorem{prop}[theorem]{Proposition}

\usepackage[round]{natbib}

\author[Bornstein, Golumbic, Santos, Souza, Szwarcfiter]{Claudson F. Bornstein\affiliationmark{1} \and Martin Charles Golumbic\affiliationmark{2}\\ \and Tanilson D. Santos\affiliationmark{1,3}  \and Uéverton S. Souza\affiliationmark{4}\thanks{This work is partially supported by Fundação de Amparo à Pesquisa do Estado do Rio de Janeiro - Brasil (FAPERJ) - grant E-26/203.272/2017; Conselho Nacional de Desenvolvimento Científico e Tecnológico – Brasil (CNPq) - grant 303726/2017-2; and Coordenação de Aperfeiçoamento de Pessoal de Nível Superior – Brasil (CAPES) - Finance Code 001.} \and Jayme L. Szwarcfiter\affiliationmark{1,5}}
\title[The Complexity of Helly-$B_{1}$-EPG graph Recognition]{The Complexity of Helly-$B_{1}$-EPG graph Recognition}
\affiliation{
  Federal University of Rio de Janeiro - Brazil \\
  University of Haifa - Israel\\
  Federal University of Tocantins  - Brazil \\
  Fluminense Federal University - Brazil \\
  State University of Rio de Janeiro - Brazil}
\received{2019-6-27}

\revised{2020-5-18}

\accepted{2020-5-19}

\begin{document}
\publicationdetails{22}{2020}{1}{19}{5603}
\maketitle
\begin{abstract}
Golumbic, Lipshteyn, and Stern defined in 2009 the class of EPG graphs, the intersection graph class of edge paths on a grid. An EPG graph $G$ is a graph that admits a representation where its vertices correspond to paths in a grid $Q$, such that two vertices of $G$ are adjacent if and only if their corresponding paths in $Q$ have a common edge. If the paths in the representation have at most $k$ bends, we say that it is a  $B_k$-EPG representation. A collection $C$ of sets satisfies the Helly property when every sub-collection of $C$ that is pairwise intersecting has at least one common element.
In this paper, we show that given a graph $G$ and an integer $k$, the problem of determining whether $G$ admits a $B_k$-EPG representation whose edge-intersections of paths satisfy the Helly property, so-called Helly-$B_k$-EPG representation, is in NP, for every $k$ bounded by a polynomial function of $|V(G)|$. Moreover, we show that the problem of recognizing Helly-$B_1$-EPG graphs is NP-complete, and it remains NP-complete even when restricted to 2-apex and 3-degenerate graphs.
\end{abstract}

\section{Introduction}
An EPG graph $G$ is a graph that admits a representation in which its vertices are represented by paths of a grid $Q$, such that two vertices of $G$ are adjacent if and only if the corresponding paths have at least one common edge.

The study of EPG graphs has motivation related to the problem of VLSI design that combines the notion of edge intersection graphs of paths in a  tree with a  VLSI  grid layout model, see~\cite{golumbic2009}. The number of bends in an integrated circuit may increase the layout area, and consequently, increase the cost of chip manufacturing.
This is one of the main applications that instigate research on the EPG representations of some graph families when there are constraints on the number of bends in the paths used in the representation.
Other applications and details on circuit layout problems can be found in~\cite{bandy1990, molitor1991}.

A graph is a $ B_k$-EPG graph if it admits a representation in which each path has at most $k$ bends. As an example, Figure~\ref{fig:trianguloepgRepresentacao}(a) shows a $C_3$, Figure~\ref{fig:trianguloepgRepresentacao}(b) shows an EPG representation where the paths have no bends and Figure~\ref{fig:trianguloepgRepresentacao}(c) shows a representation with at most one bend per path.   
Consequently, $C_3$ is a $B_0$-EPG graph. More generally, $B_0$-EPG graphs coincide with interval graphs.

\begin{figure}[h]
  \centering
  \begin{tabular}{ p{3cm} p{0.7cm} p{4cm} p{0.7cm} p{4cm} }
    \includegraphics[width=2.3cm]{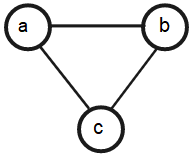} && \includegraphics[width=3.5cm]{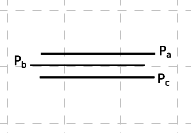} & &
    \includegraphics[width=3.5cm]{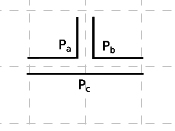}
    \\
    \footnotesize
    (a) The  graph $C_3$ && \footnotesize(b) $B_0$-EPG representation of $C_3$ (edge-clique)&& \footnotesize(c) $B_1$-EPG representation of $C_3$ (claw-clique)\\
  \end{tabular}

 \caption{The  graph $ C_3 $  and  representations without bends and with 1 bend} \label{fig:trianguloepgRepresentacao}
\end{figure}

The \emph{bend number} of a graph $G$ is the smallest $k$ for which $G$ is a $B_k$-EPG graph. Analogously, the bend number of a class of graphs is the smallest $k$ for which all graphs in the class have a $B_k$-EPG representation. Interval graphs have bend number $0$, trees have bend number $1$, see~\cite{golumbic2009}, and outerplanar graphs have bend number $2$, see~\cite{daniel2014b}. The bend number for the class of planar graphs is still open, but according to \cite{daniel2014b}, it is either $3$ or $4$.

The class of EPG graphs has been studied in several papers, such as \cite{alcon2016, Asinowski2009, cohen2014, golumbic2009, heldt2014,  martin2017,golumbic2019edge}, among others. The investigations regarding EPG graphs frequently approach characterizations concerning the number of bends of the graph representations. Regarding the complexity of recognizing $B_k$-EPG graphs, only the complexity of recognizing a few of these sub-classes of EPG graphs has been determined: $B_0$-EPG graphs can be recognized in polynomial time, since it corresponds to the class of interval graphs, see ~\cite{booth1976}; in contrast, recognizing $B_1$-EPG and $B_2$-EPG graphs are NP-complete problems, see~\cite{heldt2014} and \cite{martin2017}, respectively. 
Also, note that the paths in a $B_1$-EPG representation have one of the following shapes: $\llcorner$, $\lrcorner$, $\ulcorner$ and $\urcorner$. \cite{cameron2016edge} showed that for each $S\subset \{\llcorner, \lrcorner, \ulcorner, \urcorner\}$, it is NP-complete to determine if a given graph $G$ has a $B_1$-EPG representation using only paths with shape in $S$.

A  collection $C$ of sets satisfies the Helly property when every sub-collection of $C$ that is pairwise intersecting has at least one common element. 
The study of the Helly property is useful in diverse areas of science. We can enumerate applications in semantics, code theory, computational biology, database, image processing, graph theory, optimization, and linear programming, see \cite{dourado2009}.

The Helly property can also be applied to the $B_k$-EPG representation problem, where each path is considered a set of edges. A graph $G$ has a  Helly-$B_k$-EPG representation if there is a $B_k$-EPG representation of $G$ where each path has at most $k$ bends, and this representation satisfies the Helly property. Figure~\ref{fig:envelopeRepresentacoes}(a) presents two $B_1$-EPG representations of a graph with five vertices.  Figure~\ref{fig:envelopeRepresentacoes}(b)   illustrates 3 pairwise intersecting paths ($P_{v_1}, P_{v_2}, P_{v_5}$), containing a common edge, so it is a Helly-$B_1$-EPG representation. In Figure~\ref{fig:envelopeRepresentacoes}(c), although the three paths are pairwise intersecting, there is no common edge in all three paths, and therefore they do not satisfy the Helly property.

The Helly property related to EPG representations of graphs has been studied in~\cite{golumbic2009} and~\cite{golumbic2013}. 

Let $\cal {F}$ be a family of subsets of some universal set $U$, and $h\geq 2$ be an integer.  Say that $\cal{F}$ is $h$-{\it intersecting} when every group of $h$ sets of $\cal {F}$ intersect. The {\it core} of $\cal {F}$, denoted by $core(\cal F)$, is the intersection of all sets of $\cal {F}$. The family $\cal{F}$ is $h$-{\it Helly} when every $h$-intersecting subfamily $\cal{F'}$ of $\cal{F}$ satisfies $core(\cal{F'}) \neq \emptyset$, see e.g. \cite{D76}. On the other hand, if for every subfamily $\cal{F'}$ of $\cal{F}$, there are $h$ subsets whose core equals the core of  $\cal {F'}$, then $\cal {F}$ is said to be {\it strong} $h$-{\it Helly}.
Note that the Helly property that we will consider in this paper is precisely the property of being 2-Helly. 

The  {\it Helly number} of the family $\cal{F}$ is the least integer $h$, such that $\cal{F}$ is $h$-Helly. Similarly, the {\it strong Helly number} of $\cal{F}$ is the least $h$, for which  $\cal{F}$ is strong $h$-Helly. It also follows that the strong Helly number of $\cal{F}$ is at least equal to its Helly number. In~\cite{golumbic2009} and~\cite{golumbic2013}, they have determined the strong Helly number of $B_1$-EPG graphs. 

\begin{figure}[h]
  \centering
  \begin{tabular}{ p{3.2cm} p{4.5cm} p{4.5cm} }
    \centering \includegraphics[width=3cm]{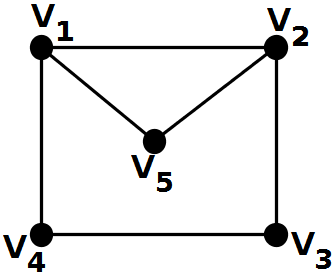} & \includegraphics[width=4cm]{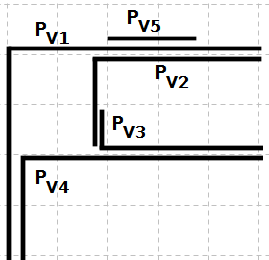} & \includegraphics[width=4cm]{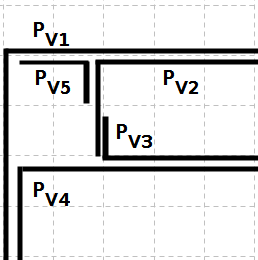}
    \\
    \footnotesize \centering (a) A  graph with 5 vertices & \footnotesize(b) $B_1$-EPG representation that satisfies the Helly property & \footnotesize (c) $B_1$-EPG representation that does  not satisfy the Helly property  \\

  \end{tabular}
\caption{A  graph with 5 vertices in (a) and some single bend representations: Helly in (b) and not Helly in (c)} \label{fig:envelopeRepresentacoes}
\end{figure}


Next, we describe some terminology and notation.

The term \emph{grid} is used to denote the Euclidean space of integer orthogonal coordinates. Each pair of integer coordinates corresponds to a \emph{point} (or vertex) of the grid. The \emph{size} of a grid is its number of points. The term \emph{edge of the grid} will be used to denote a pair of vertices that are at a distance one in the grid. Two edges $e_1$ and $e_2$ are \emph{consecutive edges} when they share exactly one point of the grid.
 A (simple) path in the grid is as a sequence of distinct edges $e_1, e_2, \leq, e_m$,  where consecutive edges are adjacent, i.e., contain a common vertex, whereas non-consecutive edges are not adjacent.  In this context, two paths only intersect if they have at least a common edge. The first and last edges of a path are called \emph{extremity edges}.
  
The \emph{direction of an edge} is vertical when the first coordinates of its vertices are equal, and is horizontal when the second coordinates are equal. A \emph {bend} in a path is a pair of consecutive edges $ e_1, e_2 $ of that path, such that the directions of $ e_1$ and $ e_2$ are different. When two edges $ e_1$ and $e_2 $ form a bend, they are called \emph { bend edges}. A \emph {segment} is a set of consecutive edges with no bends. 
Two paths are said to be \emph{edge-intersecting}, or simply  \emph{intersecting} if they share at least one edge. Throughout the paper, any time we say that two paths intersect, we mean that they edge-intersect. If every path in a representation of a graph $G$ has at most $k$ bends, we say that this graph $G$ has a \emph{$B_k$-EPG} representation. When $k = 1$ we say that this is a \emph{single bend} representation.

\medskip

In this paper, we study the Helly-$B_k$-EPG graphs. First, we show that every graph admits an EPG representation that is Helly, and present a characterization of Helly-$B_1$-EPG representations. Besides, we relate Helly-$B_1$-EPG graphs with L-shaped graphs, a natural family of subclasses of $B_1$-EPG. Finally, we prove that recognizing Helly-$B_k$-EPG graphs is in NP, for every fixed $k$. Besides, we show that recognizing Helly-$B_1$-EPG graphs is NP-complete, and it remains NP-complete even when restricted to 2-apex and 3-degenerate graphs.

The rest of the paper is organized as follows. In Section~\ref{sec:prelim}, we present some preliminary results, we show that every graph is a Helly-EPG graph, present a characterization of Helly-$B_1$-EPG representations, and relate Helly-$B_1$ EPG with L-shaped graphs. In Section~\ref{sec:NPpert}, we discuss the NP-membership of {\sc Helly-$B_k$ EPG Recognition}. In Section~\ref{sec:sectionDispositivoClausula}, we present the NP-completeness of recognizing Helly-$B_1$-EPG graphs.

\section{Preliminaries}\label{sec:prelim}

The study starts with the following lemma.
 
 \begin{lemma}[\cite{golumbic2009}] \label{lem:todoGrafoEpg}
 Every graph is an EPG graph.
 \end{lemma}
 
 We show that this result extends to Helly-EPG graphs.
 
 \begin{lemma}\label{lem:todoGrafoEpgHelly}
 Every graph is a Helly-EPG graph.
 \end{lemma}

\begin{proof}
Let $G$ be a graph with $n$ vertices $v_1, v_2, \dots, v_n$ and $\mu$ maximal cliques $C_1, C_2, \dots , C_{\mu }$. We construct a Helly-EPG representation of $G$ using a $\mu +1\times \mu +1$ grid $Q$. 
Each maximal clique $C_i$ of $G$ is mapped to an edge of $Q$ as follow: 
\begin{itemize}
    \item if $i$ is even then the maximal clique $C_i$ is mapped to the edge in column $i$ between rows $i$ and $i+1$;
    \item if $i$ is odd then the maximal clique $C_i$ is mapped to the edge in row $i$ between columns $i$ and $i+1$.
\end{itemize}

The following describes a descendant-stair-shaped construction for the paths.
  
Let $v_l \in V(G)$ and $C_i$ be the first maximal clique containing $v_l$ according to the increasing order of their indices. If $i$ is even (resp. odd) the path $P_l$ starts in column $i$ (resp. in row $i$), in the point $(i,i)$. Then $P_l$ extends to at least the point $(i+1, i)$ (resp. $(i, i+1)$) proceeding to the until the row (resp. column) corresponding to the next maximal clique of the sequence, say $C_{j}$, containing $v_l$.
At this point, we bend $P_l$, which goes to the point $(j,j)$ and repeat the process previously described. 
Figure~\ref{fig:gradeDemonstracao} shows the Helly-EPG representation of the octahedral graph $O_3$, according to the construction previously described.

By construction, each path travels only rows and columns corresponding with maximal cliques containing its respective vertex. And, every path crosses the edges of the grid to which your maximal cliques were mapped. Thus, the previously described construction results in an EPG representation of $G$, which is Helly since every set ${\mathcal P}$ of paths representing a maximal clique has at least one edge in its core.
\end{proof}
 
\begin{figure}[ht]
  \centering
  \begin{tabular}{c c c c c }
    \includegraphics[width=3.5cm]{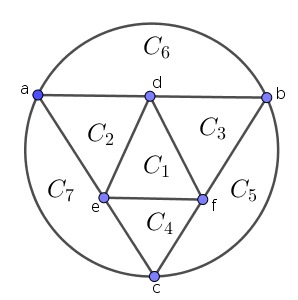} &&   
    \includegraphics[width=9cm]{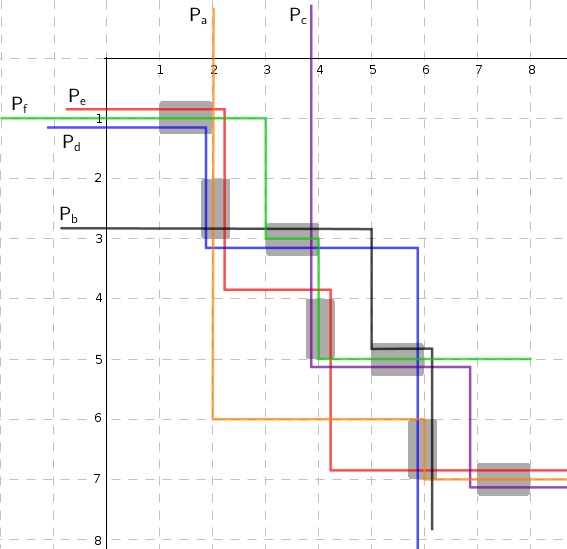}\\
    {\footnotesize (a) The octahedral graph $O_3$\vspace{.1cm}}&& 
    {\footnotesize (b) A Helly-EPG representation of the graph $O_3$} 
  \end{tabular}
  \caption{Helly-EPG representation of the graph $O_3$ according to the construction of Lemma~\ref{lem:todoGrafoEpgHelly}. The paths have been extended to the first/last row or column to improve the presentation.}\label{fig:gradeDemonstracao}
\end{figure} 
 
\begin{definition}
The \emph{Helly-bend number} of a graph $G$, denoted by $b_H(G)$, is the smallest $k$ for which $G$ is a Helly-$B_k$-EPG graph. Also, the bend number of a graph class ${\mathcal C}$ is the smallest $k$ for which all graphs in ${\mathcal C}$ have a $B_k$-EPG representation.
\end{definition}
 
\begin{corollary}\label{cor:maxCliques}
For every graph $G$ containing $\mu$ maximal cliques, it holds that $b_H(G)\leq \mu -1$. 
\end{corollary}
\begin{proof}
From the construction presented in Lemma~\ref{lem:todoGrafoEpgHelly}, it follows that any graph admits a Helly-EPG representation where its paths have a descendant-stair shape. Since the number of bends in such a stair-shaped path is the number of maximal cliques containing the represented vertex minus one, it holds that $b_H(G)\leq \mu -1$ for any graph $G$.
\end{proof}


Next, we examine the $B_1$-EPG representations of a few graphs that we employ in our constructions.

\medskip

Given an EPG representation of a graph $G$, for any grid edge $e$, the set of paths containing $e$ is a clique in $G$; such a clique is called an edge-clique. A claw in a grid consists of three grid edges meeting at a grid point. The set of paths that contain two of the three edges of a claw is a clique; such a clique is called a claw-clique, see~\cite{golumbic2009}. Fig.~\ref{fig:trianguloepgRepresentacao} illustrates an edge-clique and a claw-clique.

\begin{lemma}[\cite{golumbic2009}]\label{edge-claw-clique} 
Consider a $B_1$-EPG representation of a graph $G$. Every clique in $G$ corresponds to either an edge-clique or a claw-clique.
\end{lemma}

Next, we present a characterization of Helly-$B_1$-EPG representations.

\begin{lemma}\label{caracterization}
A $B_1$-EPG representation of a graph $G$ is Helly if and only if each clique of $G$ is represented by an edge-clique, i.e., it does not contain any claw-clique.
\end{lemma}
\begin{proof}
Let $R$ be a $B_1$-EPG representation of a graph $G$. It is easy to see that if $R$ has a claw-clique, it does not satisfy the Helly property. Now, suppose that $R$ does not satisfy the Helly property. Thus it has a set ${\mathcal P}$ of pairwise intersecting paths having no common edge. Note that the set ${\mathcal P}$ represents a clique of $G$, and by Lemma~\ref{edge-claw-clique}, every clique in $G$ corresponds to either an edge-clique or a claw-clique. Since ${\mathcal P}$ represents a clique, but its paths have no common edge, then it has a claw-clique. 
\end{proof}

Now, we consider EPG representations of $C_4$.

\begin{definition} \label{defi:tortasFrame}
Let $ Q $ be a grid and let $ (a_1, b),$ $(a_2, b),$ $(a_3, b),$ $(a_4, b)$ be a 4-star as depicted in Figure~\ref{fig:piesInGrid}(a). Let $ \mathcal{P} = \{P_1, \dots , P_4\}$ be a collection of distinct paths each containing exactly two edges of the $4$-star.
\begin{itemize}
\item A \emph{true pie} is a representation where each $P_i$ of $ \mathcal{P} $ forms a bend in $b$.

\item A \emph {false pie} is a representation where two of the paths $P_i$ do not contain bends, while the remaining two do not share an edge. 
\end{itemize}
\end{definition}

Fig.~\ref{fig:piesInGrid} illustrates true pie and false pie representations of a $C_4$.

\begin{definition} \label{defi:tortasFrame2}
 Consider a rectangle of any size with 4 corners at points $ (x_1, y_1);$ $(x_2, y_1);$ $(x_2, y_2);$ $(x_1, y_2) $, positioned as in  Fig.~\ref{fig:frameInGrid}(a). 
 \begin{itemize}
 \item A \emph{frame} is a representation containing 4 paths $\mathcal{P} =  \{ P_1, \dots, P_4\} $, each having a bend in a different corner of a rectangle, and such that the  sub-paths $ P_1 \cap P_2, P_1 \cap P_3, P_2 \cap P_4, P_3 \cap P_4 $ share at least one edge. While $P_1 \cap P_4 $ and $ P_2 \cap P_3$ are empty sets.
 
 \item A square-frame is a frame where $P_1$, $P_2$, $P_3$ and $P_4$ have respectively point of bend $ (x_1, y_1),$ $(x_2, y_1),$ $(x_1, y_2)$ and $(x_2, y_2)$, and are of the shape $\llcorner$, $\lrcorner$, $\ulcorner$ and $\urcorner$.  (see Fig.\ref{fig:frameInGrid})
 \end{itemize}
\end{definition}

Fig.~\ref{fig:frameInGrid} illustrates some frame representations of a $C_4$.

\begin{figure}[htb]
  \centering
  \begin{tabular}{c c c c c }
    \includegraphics[width=3.5cm]{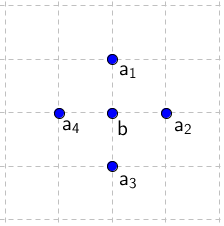}    
    & &\includegraphics[width=3.5cm]{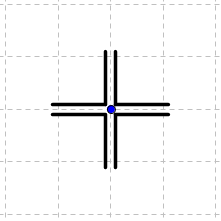} 
    & &
 \includegraphics[width=3.5cm]{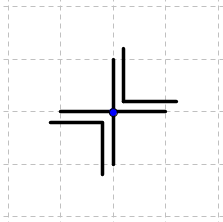} \\
    {\footnotesize (a) 4-star in grid}  & &  {\footnotesize (b) True pie} & & {\footnotesize (c) False pie} 
  \end{tabular}
  \caption{$B_{1}$-EPG representation of the induced cycle of size 4 as pies with emphasis in center $b$}\label{fig:piesInGrid}
\end{figure} 

\begin{figure}[htb]
  \centering
  \begin{tabular}{c c c c c }
    \includegraphics[width=3.5cm]{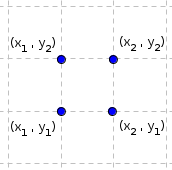}    
    & &
   \includegraphics[width=3.5cm]{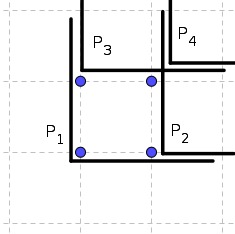} 
     & &
   \includegraphics[width=3.5cm]{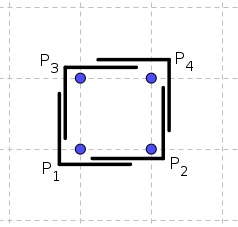} \\
   {\footnotesize (a) Points of the coordinates of bends of a frame}  
   & & {\footnotesize (b) An example of a frame} 
   & & {\footnotesize (c) A square-frame} 
  \end{tabular}
  \caption{$B_{1}$-EPG representation of the induced cycle of size 4 as frame}\label{fig:frameInGrid}
\end{figure} 

\begin{lemma}[\cite{golumbic2009}]\label{lem:representacaoC4}
Every  $C_4$ that is an induced subgraph of a graph $ G $ corresponds, in any representation, to a true pie, a false pie, or a frame.
\end{lemma}

The following is a claim of~\cite{heldt2014} which a reasoning can be found in~\cite{Asinowski2009}.

\begin{lemma}[\cite{daniel2014b} and \cite{Asinowski2009}]\label{fact:k24facts}
In every single bend representation of a $K_{2,4}$, the path representing each vertex of the largest part has its bend in a false pie.
\end{lemma} 

By creating four $K_{2,4}$ and identifying a vertex of the largest part of each one to a distinct vertex of a $C_4$, we construct the graph we called bat graph (see Fig~\ref{fig:grafoQ}). Regarding to such a graph, the following holds.

\begin{figure}[htb]
  \centering
  \begin{tabular}{c c c c c }
    \includegraphics[width=5.5cm]{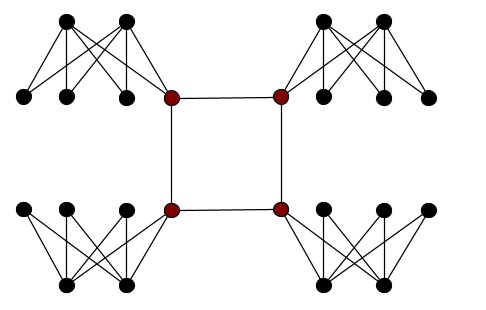}    
    & &
   \includegraphics[width=8cm]{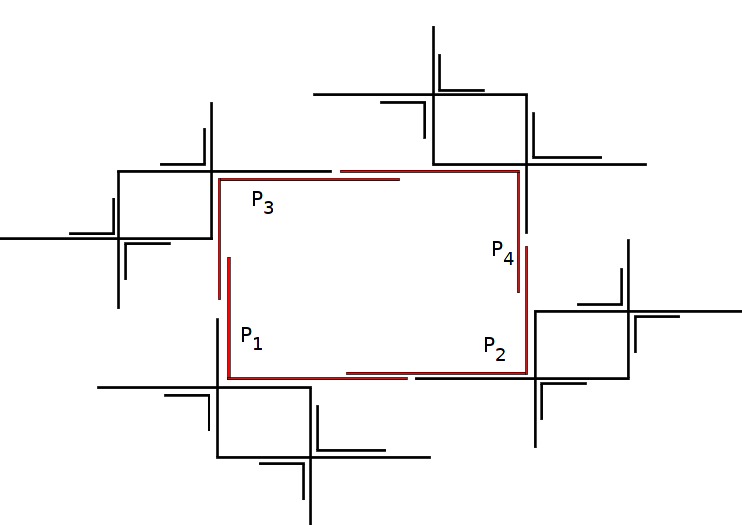}
  \end{tabular}
  \caption{A bat graph $G$ and a Helly-$B_1$-EPG representation of $G$.}\label{fig:grafoQ}
\end{figure} 

\begin{corollary}\label{batgraph}
In every single bend representation of the bat graph, $G$ presented in Fig.~\ref{fig:grafoQ}, the $C_4$ that is a transversal of all $K_{2,4}$ is represented by a square-frame.
\end{corollary}
\begin{proof}
By Lemma~\ref{fact:k24facts}, it follows that in every single bend representation of the bat graph, each path representing a vertex of the $C_4$ (transversal to all $K_{2,4}$) has its bend in a false pie in which paths represent vertices of a $K_ {2,4}$ (Fig.~\ref{fig:grafoQ2} illustrates a $B_1$-EPG representation of a $K_ {2,4}$). Thus, the intersection of two paths representing vertices of this $C_4$ does not contain any edge incident to a bend point of such paths, which implies that such a $C_4$ must be represented by a frame (see Lemma~\ref{lem:representacaoC4}). Note that for each path of the frame, we have four possible shapes ($\llcorner$, $\lrcorner$, $\ulcorner$, and $\urcorner$). Let $P_1$ be the path having the bottom-left bend point, $P_2$ be the path having the bottom-right bend point, $P_3$ be the path having the top-left bend point and $P_4$ be the path having the top-right bend point. Note that to prevent $P_2$ and $P_3$ from containing edges incident at the bend point of $P_1$, the only shape allowed for $P_1$ is $\llcorner$. Similarly, the only shape allowed for $P_2$ is $\lrcorner$ as well as for $P_3$ is $\ulcorner$ and for $P_4$ is $\urcorner$. Thus, the $C_4$ is represented by a square-frame.
\end{proof}

\begin{figure}[htb]
  \centering
  \begin{tabular}{c c c c c }
   \includegraphics[width=5cm]{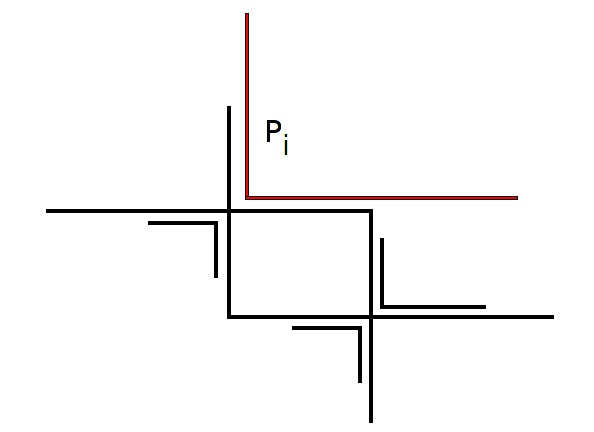}
  \end{tabular}
  \caption{Helly-$B_1$-EPG representation of a $K_{2,4}$.}\label{fig:grafoQ2}
\end{figure} 

\begin{definition}
A $B_k$-EPG representation is \emph{minimal} 
when its set of edges does not properly contain another $B_k$-EPG representation. 
\end{definition}

The \textit{octahedral} graph is the graph containing 6 vertices and 12 edges, depicted  in Figure~\ref{fig:octaedro}(a). Next, we consider representations of the octahedral graph.

The next lemma follows directly from the discussion presented in~\cite{heldt2014}.

\begin{lemma}\label{lem:octaedronaohelly}
Every minimal $B_1$-EPG representation of the octahedral graph $O_3$ has the same shape.
\end{lemma}

\begin{proof}
Note that the octahedral graph $O_3$ has an induced $C_4$ such that the two vertices 
of the octahedral graph that are not in such a cycle are false twins whose neighborhood contains the vertices of the induced $C_4$. 

If in an EPG representation of the graph $O_3$ such a $C_4$ is represented as a frame, then no single bend path can simultaneously intersect the four paths representing the vertices of the induced $C_4$. Therefore, we conclude that the frame structure cannot be used to represent such a $C_4$ in a $B_1$-EPG representation of the $O_3$. Now, take a $B_1$-EPG representation of such a $C_4$ shaped as a true pie or false pie.  By adding the paths representing the false twin vertices, which are neighbors of all vertices of the $C_4$, in both cases (from a true or false pie), we obtain representations with the shape represented in Fig.~\ref{fig:octaedro}(b). 
 \end{proof}

 \begin{figure}[h]
  \centering
  
  \begin{tabular}{@{}c@{} p{1.5cm} @{}c@{} }
   \centering \includegraphics[width=2.5cm]{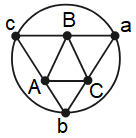} & &\includegraphics[width=2.9cm]{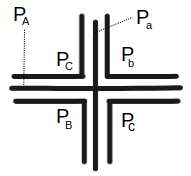}  \\[\abovecaptionskip]
    \footnotesize \centering (a) The octahedral  graph $O_3$  & &  \footnotesize(b) $B_1$-EPG representation of the graph $O_3$
  \end{tabular}

 \caption{The octahedral graph $O_3$ graph and its  $B_1$-EPG representation}\label{fig:octaedro}
\end{figure}
 
\subsection{Subclasses of $B_1$-EPG graphs}

By Lemma ~\ref{lem:octaedronaohelly}, every minimal $B_1$-EPG representation of the octahedral graph $O_3$ has the same shape, as depicted in Fig.~\ref{fig:octaedro}(b). 
Since in any representation of the graph $O_3$ there is always a triple of paths that do not satisfy the Helly property, paths $P_{a}, P_{b} $ and $P_{c}$ in the case of Fig.~\ref{fig:octaedro}(b), it holds that $O_3 \notin$ Helly-$B_1$ EPG, which implies that the class of Helly-$B_1$-EPG graphs is a proper subclass of $B_1$-EPG.

Also, $B_0$-EPG and Helly-$B_0$-EPG graphs coincide. Hence, Helly-$B_0$ EPG can be recognized in polynomial time, see \cite{booth1976}.


In a $B_1$-EPG representation of a graph, the paths can be of the following four shapes: $\llcorner$, $\lrcorner$, $\ulcorner$ and $\urcorner$. \cite{cameron2016edge} studied $B_1$-EPG graphs whose paths on the grid belong to a proper subset of the four shapes. If $S$ is a subset of $\{\llcorner, \lrcorner, \ulcorner, \urcorner\}$, then $[S]$ denotes the class of graphs that can be represented by paths whose shapes belong to $S$, where zero-bend paths are considered to be degenerate $\llcorner$'s. They consider the natural subclasses of $B_1$-EPG: $[\llcorner], [\llcorner, \ulcorner], [\llcorner, \urcorner]$ and $[\llcorner, \ulcorner, \urcorner]$, all other subsets are isomorphic
to these up to 90 degree rotation. \cite{cameron2016edge}  showed that recognizing each of these classes is NP-complete.

The following shows how these classes relate to the class of Helly-$B_1$-EPG graphs.

\begin{figure}[H]	
\center
\includegraphics[width=8.5cm]{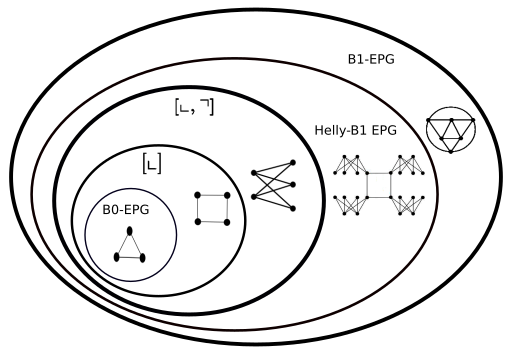} 
\caption{Hierarchical diagram of some EPG classes}
\label{fig:diagramaEPG}
\end{figure}

\begin{theorem}\label{theo:HellyLShaped}
$[\llcorner]\subsetneq [\llcorner, \urcorner]\subsetneq$~Helly-$B_1$ EPG, and Helly-$B_1$ EPG is incomparable with $[\llcorner, \ulcorner]$ and $[\llcorner, \ulcorner, \urcorner]$.
\end{theorem}
\begin{proof}
\cite{cameron2016edge} showed that $[\llcorner]\subsetneq [\llcorner, \urcorner]$. Also, it is easy to see that $\llcorner$’s and $\urcorner$’s cannot form a claw-clique, thus, by Lemma~\ref{caracterization}, it follows that $[\llcorner, \urcorner]\subseteq$~Helly-$B_1$ EPG. In order to observe that $[\llcorner, \urcorner]$ is a proper subclass of Helly-$B_1$ EPG, it is enough to analyze the bat graph (see Fig.~\ref{fig:grafoQ}): by Corollary~\ref{batgraph} it follows that any $B_1$-EPG representation of a bat graph contains a square-frame, thus it is not in $[\llcorner, \urcorner]$. In addition, the bat graph is bipartite which implies that any $B_1$-EPG representation of that graph does not contain claw-cliques and therefore is Helly.

Now, it remains to show that Helly-$B_1$ EPG is incomparable with $[\llcorner, \ulcorner]$ and $[\llcorner, \ulcorner, \urcorner]$. Again, since any $B_1$-EPG representation of a bat graph contains a square-frame, the bat graph is a Helly-$B_1$-EPG graph that is not in $[\llcorner, \ulcorner, \urcorner]$. On the other hand, the $S_3$ (3-sun) is a graph in $[\llcorner, \ulcorner]$ such that any of its $B_1$-EPG representations have a claw-clique, see~Observation 7 in~\cite{cameron2016edge}. Therefore, $S_3$ is a graph in $[\llcorner, \ulcorner]$ that is not Helly-$B_1$ EPG.
\end{proof}

Figure~\ref{fig:diagramaEPG} depicts examples of graphs of the classes $B_0$-EPG, $[\llcorner]$, $[\llcorner, \urcorner]$, Helly-$B_1$ EPG, and $B_1$-EPG that distinguish these classes.

It is known that recognizing $[\llcorner]$, $[\llcorner, \urcorner]$, and $B_1$-EPG are NP-complete while recognizing $B_0$-EPG and EPG graphs can be done in polynomial time (c.f.~\cite{booth1976}, \cite{heldt2014}, and \cite{cameron2016edge}).

In this paper, we show that it is NP-complete to recognize Helly-$B_1$-EPG graphs.

\section{Membership in NP}\label{sec:NPpert} 

The {\sc Helly-$B_k$ EPG recognition} problem can be formally described as follows.

\begin{table}[h!]
\centering
\begin{tabular}{ll}
\hline \hline
\multicolumn{2}{c}{\sc Helly-$B_k$ EPG Recognition}                         \\ \hline \hline 
\emph{Input}: & A graph $G$ and an integer $k \leq |V(G)|^c$, for some fixed $c$.\\
~ & ~ \\
\emph{Goal:}  & \begin{tabular}[c]{@{}p{9.5cm}}
Determine if there is a set of $k$-bend paths \\ $\mathcal{P} = \{P_1, P_2, \ldots, P_n\} $ in a grid $ Q $ 
such that:\\ 
$\bullet$ \ \ \ $u,v\in V(G)$ are adjacent in $G$ if only if $P_u,P_v$\\ \hspace{0.6cm} share an edge in $Q$; and\\
$\bullet$ \ \ \ $\mathcal{P}$ satisfies the Helly property.
\end{tabular} \\ \hline
\end{tabular}
\end{table}


A (positive) certificate for the {\sc Helly-$B_k$ EPG recognition} consists of a grid $Q$, a set $\mathcal{P}$ of $k$-bend paths of $Q$, which is in one-to-one correspondence with the vertex set $V(G)$ of $G$, such that, for each pair of distinct paths $P_i, P_j\in \mathcal{P}, P_i\cap P_j \neq \emptyset $ if and only if the corresponding vertices are adjacent in $G$. Furthermore, $\mathcal{P}$ satisfies the Helly property.

The following are key concepts that make it easier to control the size of an EPG representation. A \emph{relevant edge} of a path in a $B_k$-EPG representation is either an extremity edge or a bend edge of the path. Note that each path with at most $k$ bends can have up to $2(k + 1)$ relevant edges, and any $B_k$-EPG representation contains at most $2|\mathcal{P}|(k + 1)$ distinct relevant edges. 

To show that there is a non-deterministic polynomial-time algorithm for {\sc Helly-$B_k$ EPG recognition}, it is enough to consider as certificate a  $B_k$-EPG representation $R$ containing a collection $\mathcal{P}$ of paths, $|\mathcal{P}| = |V(G)|$, such that  each path $P_i \in \mathcal{P}$ is given by its set of relevant edges along with the relevant edges, that intersects $P_i$, of each path $P_j$ intersecting $P_i$, where $P_j \in \mathcal{P}$.  The relevant edges for each path are given in the order that they appear in the path, to make straightforward checking that the edges correspond to a unique path with at most $k$ bends.  This representation is also handy for checking that the paths form an intersection model for $G$.

To verify in polynomial time that the input is a positive certificate for the problem, we must assert the following:

\begin{enumerate}
\item[(i)] The sequence of relevant edges of a path $P_i\in \mathcal{P}$ determines $P_i$ in polynomial time; \label{it:bullet1}

\item[(ii)] Two paths $P_i, P_j \in \mathcal{P}$ intersect if and only if they intersect in some relevant edge; \label{it:bullet2}

\item[(iii)] The set $\mathcal{P}$ of relevant edges satisfies the Helly property.  \label{it:bullet3}
\end{enumerate}


The following lemma states that condition~(i) holds.

\begin{lemma}\label{lem:verify1}
Each path $P_i$ can be uniquely determined in polynomial time by the sequence of its relevant edges.
\end{lemma}



\begin{proof}
Consider the sequence of relevant edges of some path $P_i\in \mathcal{P}$. Start from an extremity edge of $P_i$. Let $t$ be the row (column) containing the last considered relevant edge. The next relevant edge $e'$ in the sequence, must be also contained in row (column) $t$. If $e'$ is an extremity edge, the process is finished, and the path has been determined. It contains all edges between the considered relevant edges in the sequence. Otherwise, if $e'$ is a bend edge, the next relevant edge is the second bend edge $e''$ of this same bend, which is contained in some column (row) $t'$. The process continues until the second extremity edge of $P_i$ is located. 

With the above procedure, we can determine in $\mathcal{O}(k\cdot |V(G)|)$ time, whether path $P_i$ contains any given edge of the grid $Q$. Therefore, the sequence of relevant edges of $P_i$ uniquely determines $P_i$.
 \end{proof}

Next, we assert property~(ii).

\begin{lemma}\label{lem:relevantEdges}
Let $\mathcal{P}$ be the set of paths in a $B_k$-EPG representation of $G$, and let $P_1, P_2\in \mathcal{P}$. Then $P_1$, $P_2$ are intersecting paths if and only if their intersection contains at least one relevant edge.
\end{lemma}







 

\begin{proof}
Assume that $P_1, P_2$ are intersecting, and we show they contain a common relevant edge. Without loss of generality, suppose $P_1, P_2$ intersect at row \textit{i} of the grid, in the  $B_k$-EPG representation $R$. The following are the possible cases that may occur:

\begin{itemize}
\item \textbf{Case 1:} Neither $P_1$ nor $P_2$ contain bends in row \textit{i}. 

Then $P_1$ and $ P_2$  are entirely contained in row \textit{i}. Since they intersect, either $P_1, P_2$  overlap, or one of the paths contains the other. In any of these situations, they intersect in a common extremity edge, which is a relevant edge.

\item \textbf{Case 2:} $P_1$ does not contain bends in \textit{i}, but $ P_2$ does.

If some bend edge of $P_2$ also belongs to $P_1$, then $P_1, P_2$  intersect in  a relevant edge. Otherwise, since $P_1, P_2$  intersect, the only possibility is that the intersection contains an extremity edge of $P_1$ or $ P_2$. Hence the paths intersect in a relevant edge.  

\item \textbf{Case 3:} Both $P_1$,  $P_2$ contain bends in \textit{i}

Again, if the intersection occurs in some bend edge of $P_1$  or $P_2$, the lemma follows. Otherwise, the same situation as above must occur: $P_1, P_2$  must intersect in an extremity edge.
 
\end{itemize}
In any of the cases, $P_1$ and $P_2$ intersect in some relevant edge.
 \end{proof}

The two previous lemmas let us check that a certificate is an actual $B_k$-EPG representation of a given graph $G$.  The next lemma says we can also verify in polynomial time that the representation encoded in the certificate is a Helly representation. Fortunately, we do not need to check every subset of intersecting paths of the representation to make sure they have a common intersection. 


\begin{lemma}\label{lem:verify3}
Let $\mathcal{P}$ be a collection of paths encoded as a sequence of relevant edges that constitute a  $B_k$-EPG representation of a graph $G$. We can verify in polynomial time if $\mathcal{P}$ has the Helly property.
\end{lemma}

\begin{proof}
Let $T$ be the set of relevant edges of $\mathcal{P}$. Consider each triple $T_i$ of edges of $T$ . Let $P_i$ be the set of paths of $\mathcal{P}$ containing at least two of the edges in the triple  $T_i$. By Gilmore's Theorem, see \cite{bergeDuchet1975}, $\mathcal{P}$ has the Helly property if an only if the subset of paths $P_i$  corresponding to each triple  $T_i$  has a non-empty intersection.  By Lemma~\ref{lem:relevantEdges}, it suffices to examine the intersections on relevant edges. Therefore a polynomial algorithm for checking if $\mathcal{P}$ has the Helly property could examine each of the subsets $P_i$, and for each relevant edge $e$ of a path in $P_i$, to compute the number of paths in $P_i$ that contain $e$. Then  $\mathcal{P}$ has the Helly property if and only if for every  $P_i$,  there exists some relevant edge that is present in all paths in $P_i$,  yielding a non-empty intersection.
 \end{proof}

\begin{corollary}\label{cor:comumAtodos}
Let ${\mathcal P'}$ be a set a pairwise intersecting paths in a Helly-$B_k$-EPG representation of a graph $G$. Then the intersection of all paths of  ${\mathcal P'}$ contains at least one relevant edge.
\end{corollary}

Note that the property described in Corollary~\ref{cor:comumAtodos} is a consequence of Gilmore's Theorem, see~\cite{bergeDuchet1975}, and it applies only to representations that satisfy Helly's property.

From Corollary~\ref{cor:comumAtodos}, the following theorem concerning the Helly-bend number of a graph holds.

\begin{theorem}\label{teo:lowerboundCliques}
For every graph $G$ containing $n$ vertices and $\mu$ maximal cliques, it holds that $$\frac{\mu}{2n}-1\leq b_H(G)\leq \mu -1.$$ 
\end{theorem}
\begin{proof}
The upper bound follows from Corollary~\ref{cor:maxCliques}.
For the lower bound first notice that each path with at most $k$ bends can have up to $2(k + 1)$ relevant edges, and any $B_k$-EPG representation with a set of paths $\mathcal{P}$ contains at most $2|\mathcal{P}|(k + 1)$ distinct relevant edges. Now, let $G$ be a graph with $n$ vertices, $\mu$ maximal cliques, and $b_H(G)=k$. From Corollary~\ref{cor:comumAtodos}, it follows that in a Helly-$B_k$-EPG representation of $G$ every maximal clique of $G$ contains at least one relevant edge. By maximality, two distinct maximal cliques cannot share the same edge-clique. Thus, in a Helly-$B_k$-EPG representation of $G$ every maximal clique of $G$ contains at least one distinct relevant edge, which implies that $\mu\leq 2n(k+1)$, so $\frac{\mu}{2n}-1\leq b_H(G)$.
\end{proof}

\begin{lemma}\label{lem:gridPolinomial}
Let $G$ be a (Helly-)$B_k$-EPG graph. Then $G$ admits a (Helly-)$B_k$-EPG representation on a grid of size at most $4n(k+1) \times 4n(k+1)$.
\end{lemma}
\begin{proof}
Let $R$ be a $B_k$-EPG representation of a graph $G$ on a grid $Q$ with the smallest possible size.
Let $\mathcal{P}$ be the set of paths of $R$. Note that $|\mathcal{P}|=n$.
A counting argument shows that there are at most $2|\mathcal{P}|(k+1)$ relevant edges in $R$. 
 If $Q$ has a pair of consecutive columns $c_i,c_{i+1}$ neither of which contains relevant edges of $R$, and such that there is no relevant edge crossing from $c_i$ to $c_{i+1}$, then we can contract each edge crossing from $c_i$ to $c_{i+1}$ into single vertices so as to obtain a new  $B_k$-EPG representation of $G$ on a smaller grid, which is a contradiction. An analogous argument can be applied to pairs of consecutive rows of the grid.
 Therefore the grid $Q$ is such that each pair of consecutive columns and consecutive rows of $Q$  has at least one relevant edge of $R$ or contains a relevant edge crossing it.  
  Since $Q$ is the smallest possible grid for representing $G$, then the first row and the first column of $Q$ must contain at least one point belonging to some relevant edge of $R$. 
Thus, if $G$ is $B_k$-EPG then it admits a $B_k$-EPG representation on a grid of size at most $4|\mathcal{P}|(k+1) \times 4|\mathcal{P}|(k+1)$.
Besides, by Corollary~\ref{cor:comumAtodos}, it holds that the contraction operation previously described preserves the Helly property, if any. Hence, letting $R$ be a Helly-$B_k$-EPG representation of a graph $G$ on a grid $Q$ with the smallest possible size it holds that $Q$ has size at most $4|\mathcal{P}|(k+1) \times 4|\mathcal{P}|(k+1)$.\end{proof}

Given a graph $G$ with $n$ vertices and an EPG representation $R$, it is easy to check in polynomial time with respect to $n +|R|$ whether $R$ is a $B_k$-EPG representation of $G$. By Lemma~\ref{lem:gridPolinomial}, if $G$ is a $B_k$-EPG graph then there is a positive certificate (an EPG representation) $R$ of polynomial size  with respect to $k+n$ to the question ``$G\in B_k$-EPG?''. Therefore, Corollary~\ref{BkNP} holds.

\begin{corollary}\label{BkNP}
Given a graph $G$ and an integer $k\geq 0$, the problem of determining whether $G$ is a $B_k$-EPG graph is in NP, whenever $k$ is bounded by a polynomial function of $|V(G)|$.
\end{corollary}

At this point, we are ready to demonstrate the NP-membership of {\sc Helly-$B_k$ EPG recognition}.

\medskip

\begin{theorem}\label{teo:nppertinencia}
{\sc Helly-$B_k$ EPG recognition} is in NP.
\end{theorem}
\begin{proof}
By Lemma~\ref{lem:gridPolinomial} and the fact that $k$ is bounded by a polynomial function of $|V(G)|$, it follows that the collection $\mathcal{P}$ can be encoded through its relevant edges with $n^{\mathcal{O}(1)}$ bits.

Finally, by Lemmas~\ref{lem:verify1}, \ref{lem:relevantEdges} and \ref{lem:verify3}, it follows that one can verify in polynomial-time in the size of $G$ whether $\mathcal{P}$ is a family of paths encoded as a sequence of relevant edges that constitute a Helly-$B_k$-EPG representation of a graph $G$.
\end{proof}

\section{NP-hardness}\label{sec:sectionDispositivoClausula}

Now we will prove that  {\sc Helly-$B_1$ EPG recognition} is NP-complete. For this proof, we follow the basic strategy described in the prior hardness proof of~\cite{heldt2014}. We set up a reduction from {\sc Positive (1 in 3)-3SAT} defined  as follows:


\begin{table}[h!]
\centering
\begin{tabular}{ll}
\hline \hline
\multicolumn{2}{c}{\sc Positive (1 in 3)-3SAT}                                \\ \hline \hline 
\emph{Input}: & \begin{tabular}[c]{@{}l@{}} A set $X$ of positive variables; a collection $C$ of clauses on $X$ such that \\ for each $c\in C$, $|c|= 3$.
\end{tabular} \\
\emph{Goal:}  & \begin{tabular}[c]{@{}l@{}} 
Determine if there is an assignment of values to the variables \\in $ X $ so that every clause in $ C $ has exactly one true literal.
\end{tabular} \\ \hline
\end{tabular}
\end{table}

{\sc Positive (1 in 3)-3SAT } is a well-known NP-complete problem (see \cite{johnson1979}, problem [L04], page 259). Also, it remains NP-complete when the incidence graph of the input CNF (Conjunctive Normal Form) formula is planar, see~\cite{mulzer2008minimum}.

Given a formula $F$ that is an instance of {\sc Positive (1 in 3)-3SAT} we will present a polynomial-time construction of a graph $ G_F$ such that $ G_F \in$ Helly-$B_1$ EPG if and only if $ F $ is satisfiable. This graph will contain an induced subgraph $ G_{C_i}$ with 12 vertices (called \emph {clause gadget}) for every clause $C_i \in \mathcal{C}$, and an induced subgraph (\emph {variable gadget}) for each variable $ x_j$, containing a special vertex  $ v_j$, plus a \emph{base gadget}  with 55 additional vertices.

We will use a graph $H$ isomorphic to the graph presented in Figure~\ref{fig:gadgetBase}, as a gadget to perform the proof. For each clause $C_i$ of $F$ of the target problem, we will have a \emph{clause gadget} isomorphic to $H$, denoted by $G_{c_i}$. 

\begin{figure}[htb]	
\center
\includegraphics[width=4.5cm]{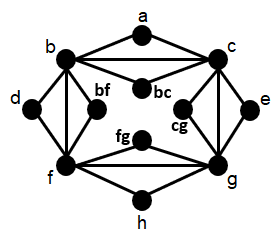}
\caption{The partial gadget graph $H$}
\label{fig:gadgetBase}
\end{figure}


The reduction of a formula $F$ from  {\sc Positive (1 in 3)-3SAT}  to a particular graph $G_F$ (where $G_F$ has a Helly-$B_1$-EPG representation if only if $F$ is satisfiable) is given below.

\begin{definition}\label{sec:reducao}
Let $F$ be a CNF-formula with variable set $\mathcal{X}$ and clause set $\mathcal{C}$ with no negative literals, in which every clause has exactly three literals. The graph $G_F$ is constructed as follows:

\begin{enumerate}
\item For each clause $C_i \in \mathcal{C}$ create a  \textit{clause gadget} $G_{C_i}$, isomorphic to  graph $H$;

\item For each variable $x_{j}\in \mathcal{X}$ create a \emph{variable vertex} $v_{j}$ that is adjacent to the vertex $a$, $e,$ or $h$ of $G_{C_i}$, when $x_{j}$ is the first, second or third variable in $C_i$, respectively;

\item For each variable vertex $v_{j}$, construct a \emph{variable gadget} formed by adding two copies of $H$, $H_1$ and $H_2$, and making $v_j$ adjacent to the vertices of the triangles $(a, b, c)$ in  $H_1$ and $H_2$.



\item Create a vertex $V$, that will be used as a vertical reference of the construction, and add an edge from $V$ to each vertex $d$ of a clause gadget;

\item Create a bipartite graph $K_{2,4}$ with a particular vertex $T$ in the largest stable set. This vertex is nominated \emph{true vertex}. Vertex $T$ is adjacent to all $v_{j}$ and also to $V$;

\item Create two  graphs isomorphic to $H$, $G_{B1}$ and $G_{B2}$. The vertex $T$ is connected to each vertex of the triangle (a,b,c) in $G_{B1}$ and $G_{B2}$;

\item Create two graphs isomorphic  to $H$, $G_{B3}$ and $G_{B4}$. The vertex $V$ is connected to each vertex of the triangle (a,b,c) in $G_{B3}$ and $G_{B4}$;

\item The  subgraph induced by the set of vertices $\{V(K_{2,4}) \cup  \{T, V\} \cup V(G_{B1}) \cup V(G_{B2}) \cup V(G_{B3}) \cup V(G_{B4})\}$ will be referred to as the  \emph{base gadget}. 
\end{enumerate}
\end{definition}

Figure~\ref{fig:exemploGrafoGF} illustrates how this construction works on a small formula. 

\begin{figure}[htb]	
\center
\includegraphics[width=6.5cm]{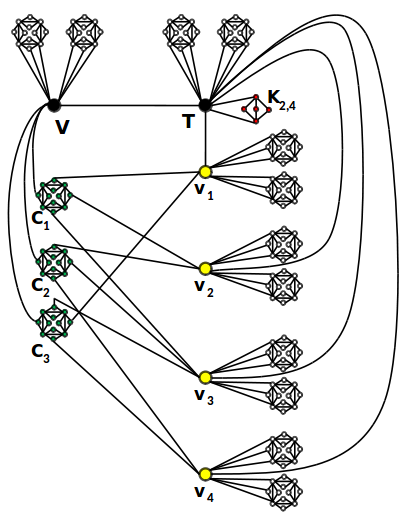}
\caption{The $G_{F}$ graph corresponding to formula $F=(x_1+ x_2+ x_3) \cdot  (x_2+ x_3+ x_4 )\cdot  (x_3 + x_1 + x_4 )$}
\label{fig:exemploGrafoGF}
\end{figure}


\begin{lemma}\label{lem:ida}
Given a satisfiable instance $F$ of {\sc Positive (1 in 3)-3SAT}, the graph $G_F$ constructed from $F$ according to Definition~\ref{sec:reducao} admits a Helly-$B_1$-EPG representation.
\end{lemma}

\begin{proof}
We will use the true pie and false pie structures to represent the \textit{clause gadgets} $ G_C$ (see Figure~\ref{fig:falseAndTruePie}), but the construction could also be done with the frame structure without loss of generality.

\begin{figure}[htb]
  \centering
  \begin{tabular}{c c c }
    \includegraphics[width=4.5cm]{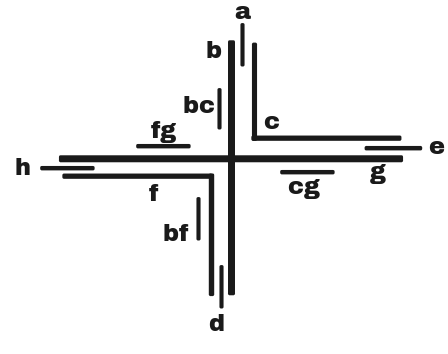}  
    & &\includegraphics[width=4.5cm]{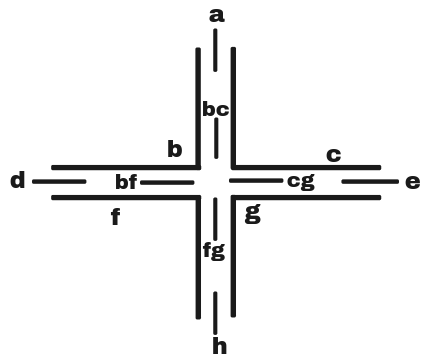} 
    \\
    {\footnotesize (a) Based in false pie}  & &  {\footnotesize(b) Based in true pie}\\
  \end{tabular}
  \caption{Single bend representations of a clause gadget isomorphic to graph $H$
  }\label{fig:falseAndTruePie}
\end{figure} 

The \textit{variable gadgets} will be represented by structures as of Figure~\ref{fig:gadgetVariavel}.

\begin{figure}[htb]	
\center
\includegraphics[width=8cm]{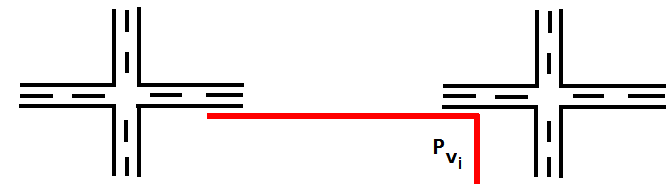}
\caption{Single bend representation of a variable gadget}
\label{fig:gadgetVariavel}
\end{figure}

The \textit{base gadget} will be represented by the structure of Figure~\ref{fig:gadgetBaseSingleBend}.

\begin{figure}[htb]	
\center
\includegraphics[width=6cm]{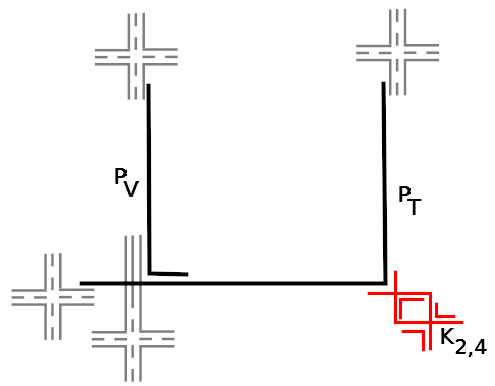}
\caption{Single bend representation of the base gadget}
\label{fig:gadgetBaseSingleBend}
\end{figure}

It is easy to see that the representations of the clause gadgets, variable gadgets, and base gadgets are all Helly-$B_1$ EPG. Now, we need to describe how these representations can be combined to construct a single bend representation $R_{G_F}$.

Given an assignment $A$ that satisfies $F$, we can construct a  Helly-$B_1$-EPG representation $R_{G_F}$. First we will fix the representation structure of the base gadget in the grid to guide the single bend representation, see Figure~\ref{fig:gadgetBaseSingleBend}. Next we will insert the variable gadgets with the following rule: if the  variable $x_i$ related to the path $P_{v_i}$ had assignment \textit{True}, then the adjacency between the path $P_{v_i}$ with $P_T$ is horizontal, and vertical otherwise. For example, for an assignment $A=\{x_1=False; x_2=False;x_3=True; x_4=False\}$  to variables of the formula $F$ that generated the gadget $G_F$ of Figure~\ref{fig:exemploGrafoGF}, it will give us a single bend representation (base gadget + variables gadget) according to the Figure~\ref{fig:gadgetBasePlusVariables}(a). 

When a formula $F$ of {\sc Positive (1-in-3)-3sat} has clauses whose format of the assignment is $(False,$ $True,$ $False)$ or $(False, False, True)$ then we will use false pie to represent these clauses. When the clause has format $(True, False, False)$, we will use true pie to represent this clause (the use of true pie in the last case is only to illustrate that the shape of the pie does not matter in the construction). To insert a \textit{ clause gadget} $G_{C}$, we introduce a horizontal line $l_{h}$ in the grid between the horizontal rows used by the paths for the two false variables in $ C $. Then we connect the path $P_{d_{c_i}}$ of $G_{C_i}$ to $P_V$ vertically using the bend of $P_{d_{c_i}}$. We introduce a vertical line $ l_{v}$ in the grid, between the vertical line of the grid used by $P_V$ and the path to the true variable in $C_i$, \textit{i.e.} between $P_V$ and the path of the true variable $x_j \in C_i$. At the point where $l_{h}$ and $l_{v}$ cross, to insert the center of the  \textit{clause gadget} as can be seen in Figure~\ref{fig:gadgetOnePie}(b). The complete construction of this single bend representation for the $G_F$ can be seen in 
Figure~\ref{fig:gadgetFormulaCompletaPies}.

Note that when we join all these representations of gadgets that form $ R_{G_F}$, we do not increase the number of bends. Then the representation necessarily is $B_1$-EPG. Let us show that it satisfies the Helly property. 

A simple way to check that $ R_{G_F} $ satisfies the Helly property is to note that the particular graph $G_F$ never forms triangles between variable, clause, and base gadgets. Thus, any triangle of $G_F$ is inside a variable, clause, or base gadget. As we only use Helly-$B_1$-EPG representations of such gadgets, $ R_{G_F} $ is a Helly-$B_1$-EPG representation of $G_F$.
 \end{proof}

Now, we consider the converse. Let $R$ be a Helly-$B_1$-EPG representation of $G_F$.



\begin{definition}
Let $H$ be the graph shown in Figure~\ref{fig:gadgetBase}, such that the 4-cycle $H[\{b, c, f, g \}]$ corresponds in $R$ to a false pie or true pie, then:

\begin{itemize}
\item the \emph{center} is the unique grid-point of this representation which is contained in every path representing 4-cycle $ \{b, c, f, g \}$; \label{lab:lab1}

\item a \emph{central ray} is an edge-intersection between two of the paths corresponding to vertices  $ b, c, f, g$, respectively.
\end{itemize}
\end{definition}

Note that every $B_1$-EPG representation of a $C_4$ satisfies the Helly property, see Lemma~\ref{lem:representacaoC4}, and triangles have $B_1$-EPG representations that satisfy the Helly property, \textit{e.g.} the one shown in Figure~\ref{fig:trianguloepgRepresentacao}(b). The graph $H$ is composed by a 4-cycle  $C_4^{H}=H[b, c, f, g]$ and eight cycles of size 3.

As $C_4^{H}$ has well known representations (see in Lemma~\ref{lem:representacaoC4}), then we can start drawing the Helly-$B_1$-EPG representation of $H$ from these structures.  Figure~\ref{fig:falsepietruepieframe} shows possible representations for $H$.


If $C_4^{H}$ is represented by a pie, then the paths $P_{b}, P_{c}, P_{f}, P_{g}$ share the center of the representation. On the other hand, if $C_4^{H}$ is represented by a frame, then the bends of the four paths correspond to the four distinct corners of a rectangle, \textit{i.e.} all paths representing the vertices of $C_4^{H}$ have distinct bend points, see~\cite{golumbic2009}.

\begin{figure}[H]
  \centering
  \begin{tabular}{p{7cm} p{7cm} }
   \centering \includegraphics[width=7cm, left]{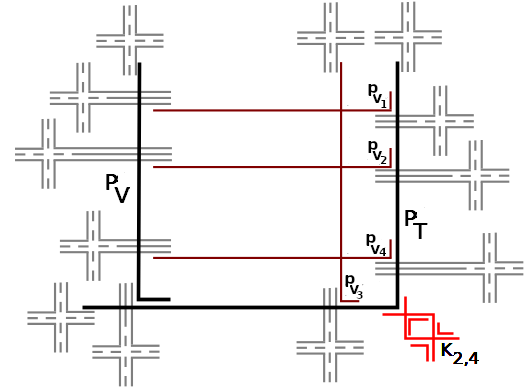} &   
  \includegraphics[width=7cm, left]{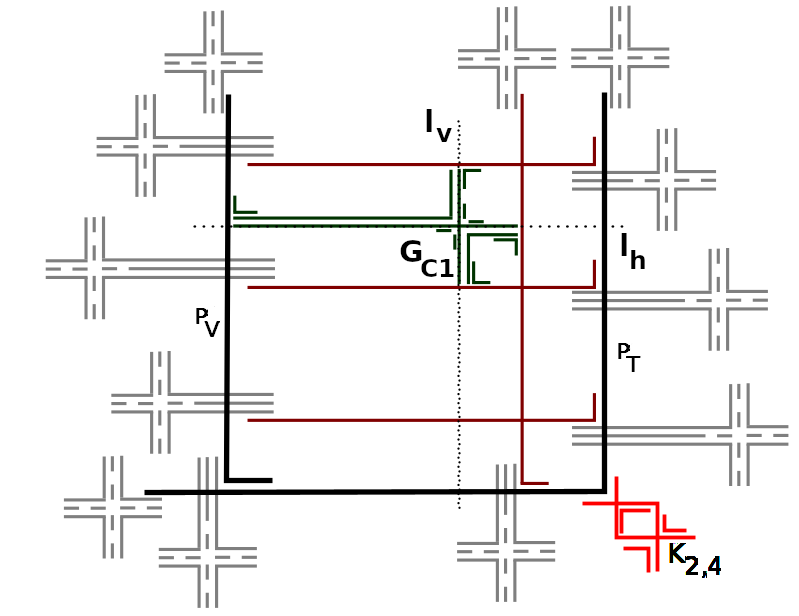}
  \\
  [\abovecaptionskip]
    \footnotesize \centering (a) Representation with omitted clause gadgets &
  \footnotesize(b) Representation with  $G_{C_1}$  associated with the clause $(x_1+x_2+x_3)$ in highlighted \\
  \end{tabular}

 \caption{Single bend representation of the base and variables gadgets associated with the assignment $x_1=False, x_2=False, x_3=True, x_4=False$} \label{fig:gadgetOnePie} \label{fig:gadgetBasePlusVariables}
\end{figure}

\begin{figure}[htb]	
\center
\includegraphics[width=10cm]{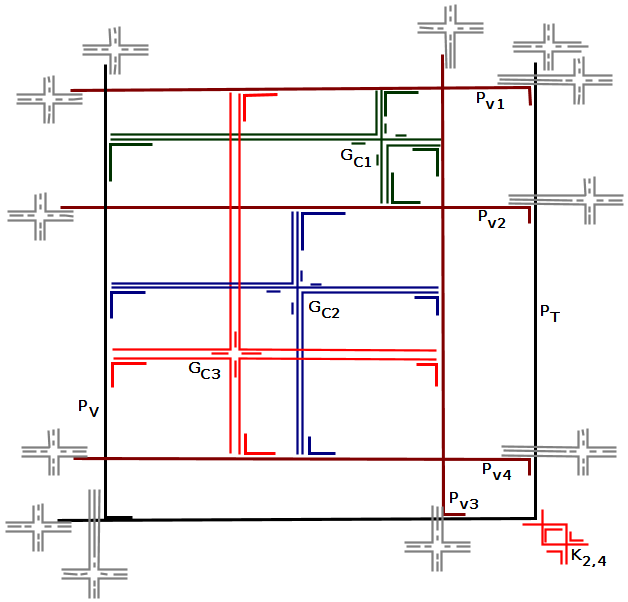}
\caption{Single bend representation of $G_F$}
\label{fig:gadgetFormulaCompletaPies}
\end{figure}

Next, we examine the use of the frame structure.


\begin{prop}\label{lem:direcoesdiferentes} 
In a frame-shaped $B_1$-EPG representation of a $C_4$, every path $P_i$ that represents a vertex of the $C_4$ intersects exactly two other paths $P_{i-1}$ and $P_{i+1}$ of the frame so that one of the intersections is horizontal and the other is vertical. 
\end{prop}


\begin{prop}\label{lem:mesmaretasuporte}
Given a Helly-$B_1$-EPG representation of a graph $G$ that has an induced $C_4$ whose representation is frame-shaped. If there is a vertex $v$ of $G$, outside the $C_4$, that is adjacent to exactly two consecutive vertices of this $C_4$, then the path representing $v$ shares at least one common edge-intersection with the paths representing both of these vertices.
\end{prop}


\begin{proof}
By assumption, $G$ has a triangle containing $v$ and two vertices of a $C_4$. Therefore the path representing $v$ shares at least one common edge intersecting with the paths representing these neighbors, otherwise the representation does not satisfy the Helly property.
 \end{proof}

By Proposition~\ref{lem:direcoesdiferentes} and Proposition~\ref{lem:mesmaretasuporte} we can conclude that for every vertex $v_i \in V(H)$ such that $v_i \neq V(C_4^{H})$, when we use a frame to represent the $C_4^{H}$, $P_{v_i}$ will have at least one common edge-intersection with the pair of paths representing its neighbors in $H$. 
Figure~\ref{fig:falsepietruepieframe}(c) presents a possible Helly-$B_1$-EPG representation of $H$. 
Note that we can apply rotations and mirroring operations while maintaining it as a Helly-$B_1$-EPG representation of $H$.

\begin{definition}
In a frame-shaped single bend representation of a $C_4$  graph, the paths that represent consecutive vertices in the $C_4$ are called \emph{con\-se\-cu\-ti\-ve paths} and the segment that corresponds to the intersection between two consecutive paths is called \emph{side intersection}.  
\end{definition}

\begin{lemma}\label{lem:2vertical2horizontal}
In any minimal single bend representation of a graph isomorphic to $H$, there are two paths in $\{P_a, P_e, P_{d}, P_{h} \}$ that have horizontal directions and the other two paths have vertical directions.
\end{lemma}




\begin{proof}
If the $C_4^{H} = [b,c,f,g]$ is  represented by a true pie or false pie, then each path of $C_4^{H}$ shares two central rays with two other paths of $C_4^{H}$, where each central ray corresponds to one pair of consecutive vertices in $C_4^{H}$.

As the vertices $a, e, d $ and $ h$ are adjacent to pairs of consecutive vertices in $C_4^{H}$ so the paths $P_{a}, P_{e}, P_{d}$ and $P_{h}$ have to be positioned in each one of the different central rays,  2 are horizontal  and 2 are vertical.

If the $C_4^{H}$ is represented by a frame, then each path of the $C_4^{H}$ has a bend positioned in the corners of the frame. In the frame, the adjacency relationship of pairs of consecutive vertices in the $C_4^{H}$ is represented by the edge-intersection of the paths that constitute the frame. Thus, since a frame has two parts in the vertical direction and two parts in the horizontal direction, then there are two paths in $\{P_{a}, P_{e}, P_{d}, P_{h}\}$ that have horizontal direction and two that have vertical direction.

Note that no additional edge is needed on the different paths by the minimality of the representation.
 \end{proof}

\begin{corollary} \label{coro:paresMesmoSegmento}
In any minimal single bend representation of a graph isomorphic to $H$, the following paths are on the same central ray or side intersection: $P_{a}$ and $P_{bc}$; $P_{e}$ and $P_{cg}$; $P_{h}$ and $P_{fg}$; $P_{d}$ and $P_{bf}$.

\end{corollary}

\begin{figure}[htb]
  \centering
  \begin{tabular}{c c c c c }
    \includegraphics[width=4cm]{falsePie}  
    & &\includegraphics[width=4cm]{truePie} 
    & &
 \includegraphics[width=4cm]{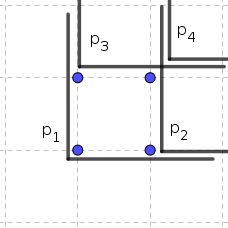} \\
    {\footnotesize (a) Based in false pie}  & &  {\footnotesize(b) Based in true pie} & & {\footnotesize (c) Based in frame} 
  \end{tabular}
  \caption{Different single bend representations of the  graph $H$ using a false pie (a), a true pie (b) and a frame (c) for representing  $C_4^{H}$}\label{fig:falsepietruepieframe}
\end{figure} 

\begin{figure}[htb]	
\center
\includegraphics[width=4cm]{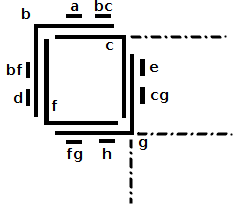}
\caption{A frame representation where the bend of dashed paths change directions}
\label{fig:outraRepresentacaoFrame}
\end{figure}


The following proposition helps us in the understanding of the NP-hardness proof.

\begin{prop}
In any Helly-$B_1$ representation of the graph $G'$, presented in Figure~\ref{fig:extremidadeDobraObstruida}(a), the path $P_{x}$ has obstructed extremities and bends.
\end{prop}

\begin{proof} 
Consider $G'$ consisting of a vertex $x$ together with two graphs, $ H_1 $ and $ H_2 $, isomorphic to $H$ and a bipartite graph $K_{2,4}$, such that: $x$ is a vertex of the largest stable set of the $K_{2,4}$; $x$ is adjacent to an induced cycle of size 3 of $H_1$, $C_3^{H_1}$ and to an induced cycle of size 3 of $H_2$, $ C_3^{H_2}$, see Figure~\ref{fig:extremidadeDobraObstruida}(a).

We know that the paths belonging to the largest stable set of a $K_{2,4}$ always will bend into a false pie, see Fact~\ref{fact:k24facts}. Since $P_{x}$ is part of the largest stable set of the $K_{2,4}$, then $P_{x}$ has an \emph {obstructed bend}, see Figure~\ref{fig:extremidadeDobraObstruida}(b). 

The vertex $x$ is adjacent to $ C_{3}^{H_1}$ and $ C_3^{H_2}$, so that its path $ P_{x} $ intersects the paths representing them.  But in a single bend representations of a graph isomorphic to $H$ there are pairs of paths that always are on some segment of a central ray or a side intersection, see Corollary~\ref{coro:paresMesmoSegmento}, and the representation of $C_{3}^{H_1}$ ( similarly $C_3^{H_2})$ has one these paths. Therefore, there is an edge in the set of paths that represent ${H_1}$ ( similarly in ${H_2}$) that has a intersection of 3 paths representing $ C_{3}^{H_1}$ (and $ C_3^{H_2}$), otherwise the representation would not be Helly. There is another different edge in the same central ray or side intersection that contains three other paths and one of them is not in the set of paths  $C_{3}^{H_1}$ ( similarly $C_3^{H_2})$. Thus in a single bend representation of $G'$, the paths that represent  $C_{3}^{H_1}$ ( similarly $C_3^{H_2})$ must intersect in a bend edge or an extremity edge of $P_{x}$, because $P_{x}$ intersects only one of the paths that are on some central ray or side intersection where  $C_{3}^{H_1}$ ( similarly $C_3^{H_2})$ is. As the bend of $G'$ is already obstructed by structure of $K_{2,4}$, then ${H_1}$ ( similarly in ${H_2}$) must be positioned at an extremity edge of $P_{x}$. This implies that $ P_{x} $ has a condition of \emph{obstructed extremities}, see Figure~\ref{fig:extremidadeDobraObstruida}(b).
\begin{figure}[h]
  \centering
  \begin{tabular}{p{6cm} p{1cm} p{6cm}}
     \includegraphics[width=5cm, center]{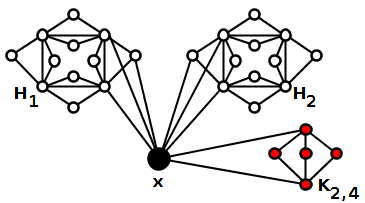} &  &\includegraphics[width=5cm, center]{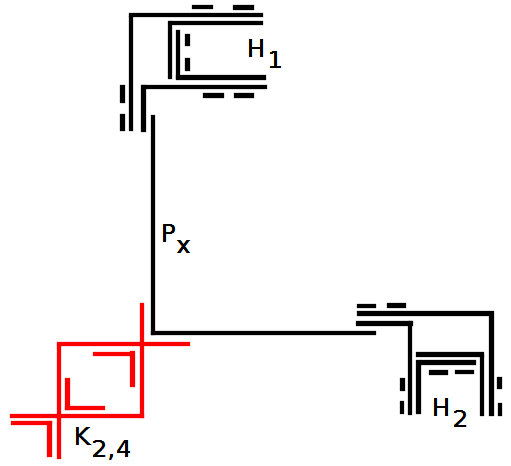}  \\
    \footnotesize \centering (a) The graph $G'$& & \footnotesize \centering (b) A $B_1$-EPG representation of $G'$
  \end{tabular}
 \caption{The sample of  obstructed extremities and bend.}\label{fig:extremidadeDobraObstruida}
\end{figure}
 \end{proof}

\begin{definition}
We say that a segment $s$ is \emph{internally contained} in a path $P_{x}$ if $s$ is contained in $P_{x}$, and it does not intersect a relevant edge of $P_{x}$. 
\end{definition}


Some of the vertices of $G_F$ have highly constrained $B_1$-EPG representations. Vertex $T$ has its bend	and both extremities	obstructed	by its neighbors in	$G_{B1}$, $G_{B2}$ and in the $K_{2,4}$ subgraphs. Vertex $V$ and each variable vertex $v_i$ must have one of its segments internally contained in $T$, and also have its extremities and bends obstructed.  Therefore, vertex $V$ and each
variable vertex has only one segment each that can be used in an EPG
representation to make them adjacent to the clause gadget. The direction of
this segment, being either horizontal	or vertical, can be used to represent
the true or false value	for the	variable.
The clause gadgets, on the other hand, are such that exactly two of its
adjacencies to the variable vertices and $V$ can be realized with a
horizontal intersection, whereas the other two must be realized with a
vertical intersection. If we consider the direction used by $V$ as a
truth assignment, we get that exactly one of the variables in each clause
will be true in	any possible representation of $G_F$. Conversely, it is	fairly straightforward to obtain a $B_1$-EPG representation for $G_F$ when given a truth assignment for the formula $F$. Therefore, Lemma~\ref{lem:volta} holds.

\begin{lemma}\label{lem:volta}
If a graph $G_F$, constructed according to Definition~\ref{sec:reducao}, admits a Helly-$B_1$-EPG representation, then the associated CNF-formula $F$ is a yes-instance of {\sc Positive (1 in 3)-3sat}.
\end{lemma}

\begin{proof}
Suppose that $G_F$ has a Helly-$B_1$-EPG representation, $R_{G_F}$.  From $R_{G_F}$ we will construct an assignment that satisfies $F$. 

First, note that in every single bend representation of a $K_{2,4}$, the path of each vertex of the largest stable set, in particular, $P_T$ (in $R_{G_F}$), has bends contained in a false pie (see Lemma~\ref{fact:k24facts}).

The vertex $T$ is adjacent to the vertices of a triangle of $G_{B1}$ and $G_{B2}$. As the $K_{2,4}$ is positioned in the bend of $P_T$, then in $R_{G_F}$ the representations of $G_{B1}$ and $G_{B2}$ are positioned at the extremities of $P_T$, see Proposition 4.3. 

Without loss of generality assume that $P_V \cap P_T$ is a horizontal segment in $R_{G_F}$.

We can note in $R_{G_F}$ that: the number of paths $P_{d}$ with segment internally contained in $P_V$ is the number of clauses in $F$; the intersection between each $P_{a}, P_{e}, P_{h}$ in the gadget clause and each path $P_{v_j}$ indicates the variables composing the clause. Thus, we can assign to each variable $ x_{j}$ the value \textit{True} if the edge intersecting $P_{v_j}$ and $P_T$ is horizontal, and \textit{False} otherwise.

In Lemma~\ref{lem:2vertical2horizontal} it was shown that any minimal $B_1$-EPG representation of a clause gadget has two paths in $\{P_{a}, P_{d}, P_{e}, P_{h}\}$ with vertical direction and the other two paths have horizontal direction. Since $P_{d}$ intersects $P_V$, it follows that in a single bend representation of $G_F$, we must connect two of these to represent a false assignment, and exactly one will represent a true assignment. Thus, from $R_{G_F}$, we construct an assignment to $F$ such that every clause has exactly one variable with a true value. 
 \end{proof}

Recall that a $B_1$-EPG representation is Helly if and only if each clique is represented by an edge-clique (and not by a claw-clique). 
Thus, an alternative way to check whether a representation is Helly is to note that all cliques are represented as edge-cliques. 

\begin{theorem}
{\sc Helly-$B_1$ EPG recognition} is NP-complete.
\end{theorem}
\begin{proof} 
By Theorem~\ref{teo:nppertinencia}, Lemma~\ref{lem:ida}, Lemma~\ref{lem:volta}.
 \end{proof}

We say that a $k$-apex graph is a graph that can be made planar by the removal of $k$ vertices. A $d$-degenerate graph is a graph in which every subgraph has a vertex of degree at most $d$. Recall that {\sc Positive (1 in 3)-3SAT} remains NP-complete when the incidence graph of the input formula is planar, see~\cite{mulzer2008minimum}. Thus, the following corollary holds.

\begin{corollary}\label{coro:2apexAnd3degenerate}
{\sc Helly-$B_1$ EPG recognition} is NP-complete on $2$-apex and $3$-degenerate graphs.
\end{corollary}

\begin{proof} 
To prove that $G_F$ is 3-degenerate, we apply the $d$-degenerate graphs recognition algorithm,  consisting of repeatedly removing the vertices of a minimum degree from the graph. Note that each vertex to be removed at each iteration of the algorithm always has a degree at most three, and therefore the graphs $G_F$ constructed according to Definition~\ref{sec:reducao} is $3$-degenerate.

Now, recall that {\sc Positive (1 in 3)-3SAT} remains NP-complete when the incidence graph of $F$ is planar, see~\cite{mulzer2008minimum}. Let $F$ be an instance of {\sc Planar Positive (1 in 3)-3SAT}, we know that the incidence graph of the formula $F$ is planar. By using the planar embedding of the incidence graph, we can appropriately replace the vertices representing variables and clauses by variables gadgets and clauses gadgets. As each variable gadget, clause gadget, and base gadget are planar, then something not planar may have arisen only from the intersection that was made between them. As the incidence graph assures that there is a planar arrangement between the intersections of the variable gadgets and clause gadgets, then from that one can construct a graph $G_F$ such that the removal of $V$ and $T$ results into a planar graph, see Figura~\ref{fig:grafoIncidenciaCompleto}. Thus $G_F$ is 2-apex.
\end{proof}

\begin{figure}[H]	
\center
\includegraphics[width=8cm]{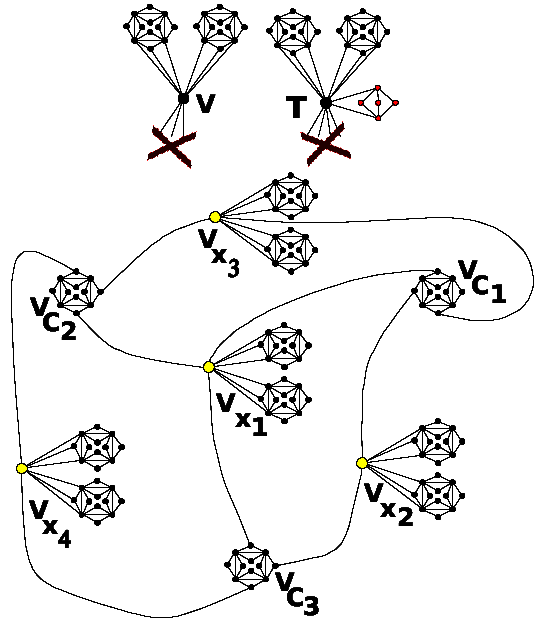}
\caption{Planar graph built from $F = (x_1 +x_2 +x_3)\cdot(x_1 +x_3 +x_4)\cdot(x_1 +x_2 +x_4)$, after removing $V$ and $T$.}
\label{fig:grafoIncidenciaCompleto}
\end{figure}

\section{Concluding Remarks}

In this paper, we show that every graph admits a Helly-EPG representation, and $\frac{\mu}{2n}-1\leq b_H(G)\leq \mu -1$. Besides, we relate Helly-$B_1$-EPG graphs with L-shaped graphs, a natural family of subclasses of $B_1$-EPG. Also, we prove that recognizing (Helly-)$B_k$-EPG graphs is in NP, for every fixed $k$. Finally, we show that recognizing Helly-$B_1$-EPG graphs is NP-complete, and it remains NP-complete even when restricted to 2-apex and 3-degenerate graphs.

Now, let $r$ be a positive integer and let $K_{2r}^-$ be the cocktail-party graph, i.e., a complete graph on $2r$ vertices with a perfect matching removed. Since $K_{2r}^-$ has $2^r$ maximal cliques, by Theorem~\ref{teo:lowerboundCliques} follows that $\frac{2^r}{4r}-1\leq b_H(K_{2r}^-)$. This implies that, for each $k$, the graph $K_{2(k+5)}^-$ is not a Helly-$B_k$-EPG graph. Therefore, as \cite{martin2017} showed that every cocktail-party graph is in $B_2$-EPG, we conclude the following.

\begin{lemma}
Helly-$B_k$-EPG $\subsetneq B_k$-EPG for each $k>0$.
\end{lemma}

The previous lemma suggests asking about the complexity of recognizing Helly-$B_k$-EPG graphs for each $k>1$. Also, it seems interesting to present characterizations for Helly-$B_k$-EPG representations similar to Lemma~\ref{caracterization} (especially for $k=2$) as well as considering the $h$-Helly-$B_k$ EPG graphs. Regarding L-shaped graphs, it also seems interesting to analyse the classes Helly-$[\llcorner, \ulcorner]$ and Helly-$[\llcorner, \ulcorner, \urcorner]$ (recall Thereom~\ref{theo:HellyLShaped}).

\bibliographystyle{abbrvnat}
\bibliography{sample-dmtcs}
\label{sec:biblio}


\end{document}